%% file: main.tex
\documentclass[letterpaper,conference]{IEEEtran}

\usepackage{cite}
\usepackage{graphicx}
\usepackage{subcaption}
\usepackage[cmex10]{amsmath}
\usepackage{algpseudocode}
\usepackage{algorithmicx}
\usepackage{array}
\usepackage{url}
\usepackage{bbm}
\usepackage{todonotes}

\usepackage{epsfig}
\usepackage{amsfonts,amssymb,multirow,bigstrut,booktabs,ctable,latexsym}
\usepackage{mathtools}
\usepackage[mathscr]{eucal}
\usepackage{verbatim}

\usepackage{amsthm}
\usepackage{algorithm}
\usepackage[normalem]{ulem}

\IEEEoverridecommandlockouts
\def\BibTeX{{\rm B\kern-.05em{\sc i\kern-.025em b}\kern-.08em
		T\kern-.1667em\lower.7ex\hbox{E}\kern-.125emX}}
\begin{document}

	\newtheorem{definition}{Definition}
	\newtheorem{lemma}{Lemma}
	\newtheorem{theorem}{Theorem}
	\newtheorem{example}{Example}
	\newtheorem{proposition}{Proposition}
	\newtheorem{remark}{Remark}
	\newtheorem{assumption}{Assumption}
	\newtheorem{corrolary}{Corrolary}
	\newtheorem{property}{Property}
	\newtheorem{ex}{EX}
	\newtheorem{problem}{Problem}
	\newcommand{\argmin}{\arg\!\min}
	\newcommand{\argmax}{\arg\!\max}
	\newcommand{\st}{\text{s.t.}}
	\newcommand \dd[1]  { \,\textrm d{#1}  }
	
	\title{Control Synthesis for Cyber-Physical Systems to Satisfy Metric Interval Temporal Logic Objectives under Timing and Actuator Attacks*
		\thanks{This work was supported by the National Science Foundation, the Office of Naval Research, and U.S. Army Research Office via Grants CNS-1656981, N00014-17-S-B001, and W911NF-16-1-0485 respectively. }
	}
	\author{
		\IEEEauthorblockN{Luyao Niu\IEEEauthorrefmark{1}, Bhaskar Ramasubramanian\IEEEauthorrefmark{2}, Andrew Clark\IEEEauthorrefmark{1}, Linda Bushnell\IEEEauthorrefmark{2}, and Radha Poovendran\IEEEauthorrefmark{2}}
		\IEEEauthorblockA{\IEEEauthorrefmark{1}Department of Electrical and Computer Engineering, Worcester Polytechnic Institute, Worcester, MA
			\\
		\IEEEauthorblockA{\IEEEauthorrefmark{2}Network Security Lab, Department of Electrical and Computer Engineering, 
						University of Washington, Seattle, WA
			\\\{lniu,aclark\}@wpi.edu, \{bhaskarr, lb2, rp3\}@uw.edu}}
	}
	

	
	\maketitle
	
	\begin{abstract}
		This paper studies the synthesis of controllers for cyber-physical systems (CPSs) that are required to carry out complex tasks that are time-sensitive, in the presence of an adversary. 
		The task is specified as a formula in metric interval temporal logic (MITL). 
		The adversary is assumed to have the ability to tamper with the control input to the CPS and also manipulate timing information perceived by the CPS.  
		In order to model the interaction between the CPS and the adversary, and also the effect of these two classes of attacks, we define an entity called a durational stochastic game (DSG). 
		DSGs probabilistically capture transitions between states in the environment, and also the time taken for these transitions. 
		With the policy of the defender represented as a finite state controller (FSC), we present a value-iteration based algorithm that computes an FSC that maximizes the probability of satisfying the MITL specification under the two classes of attacks. 
		A numerical case-study on a signalized traffic network is presented to illustrate our results.
	\end{abstract}
	
	\input{intro}
	\input{prelim.tex}\input{formulation.tex}
\input{sol.tex}\input{simulation.tex}
\input{relatedwork.tex}
	\input{conclusion.tex}
	\begin{figure*}[t!]
		\centering
		\begin{subfigure}{.78\columnwidth}
			\includegraphics[width=\columnwidth]{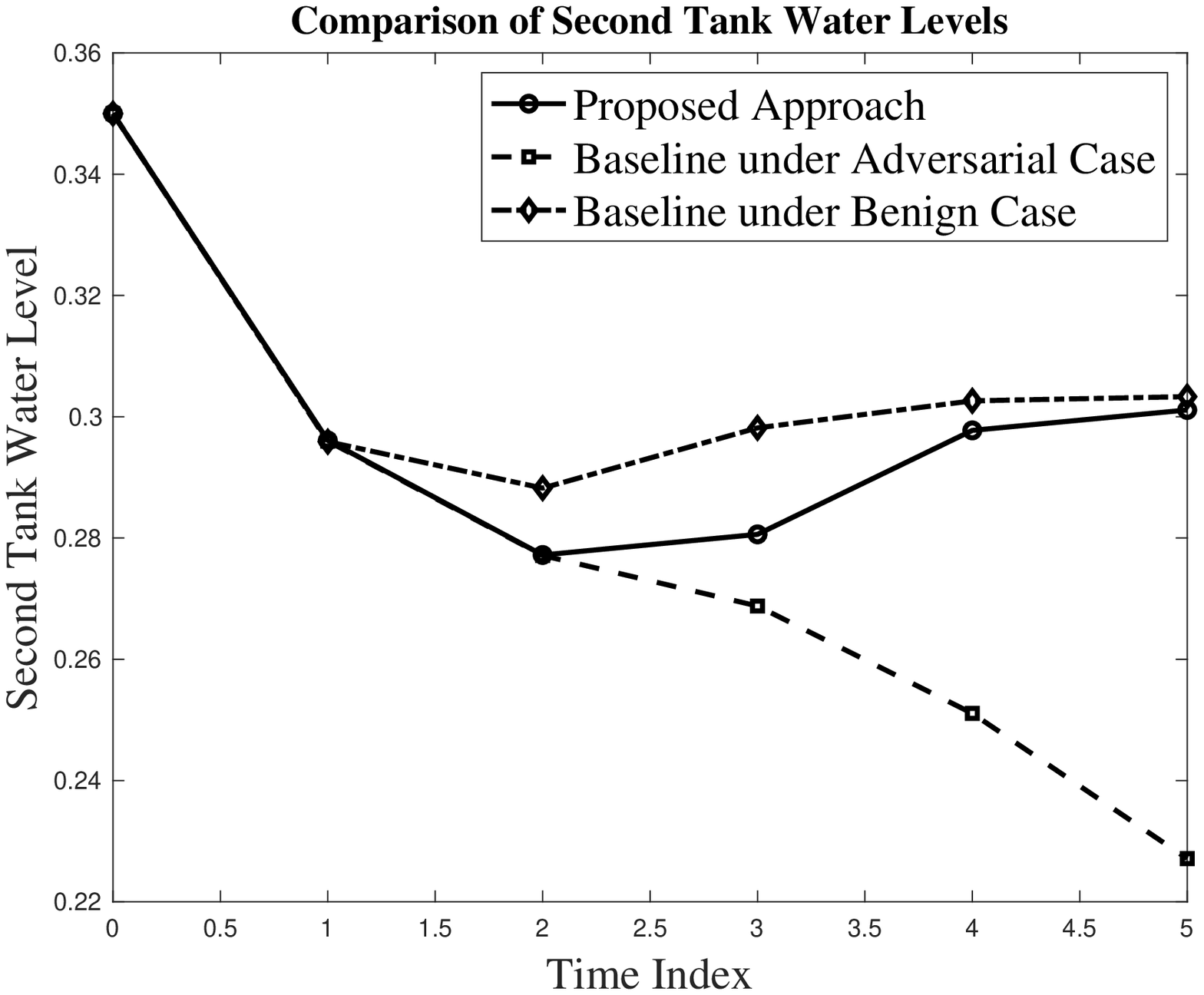}
			\subcaption{}
			\label{fig:tank2}
		\end{subfigure}\hfill
		\begin{subfigure}{.78\columnwidth}
			\includegraphics[width=\columnwidth]{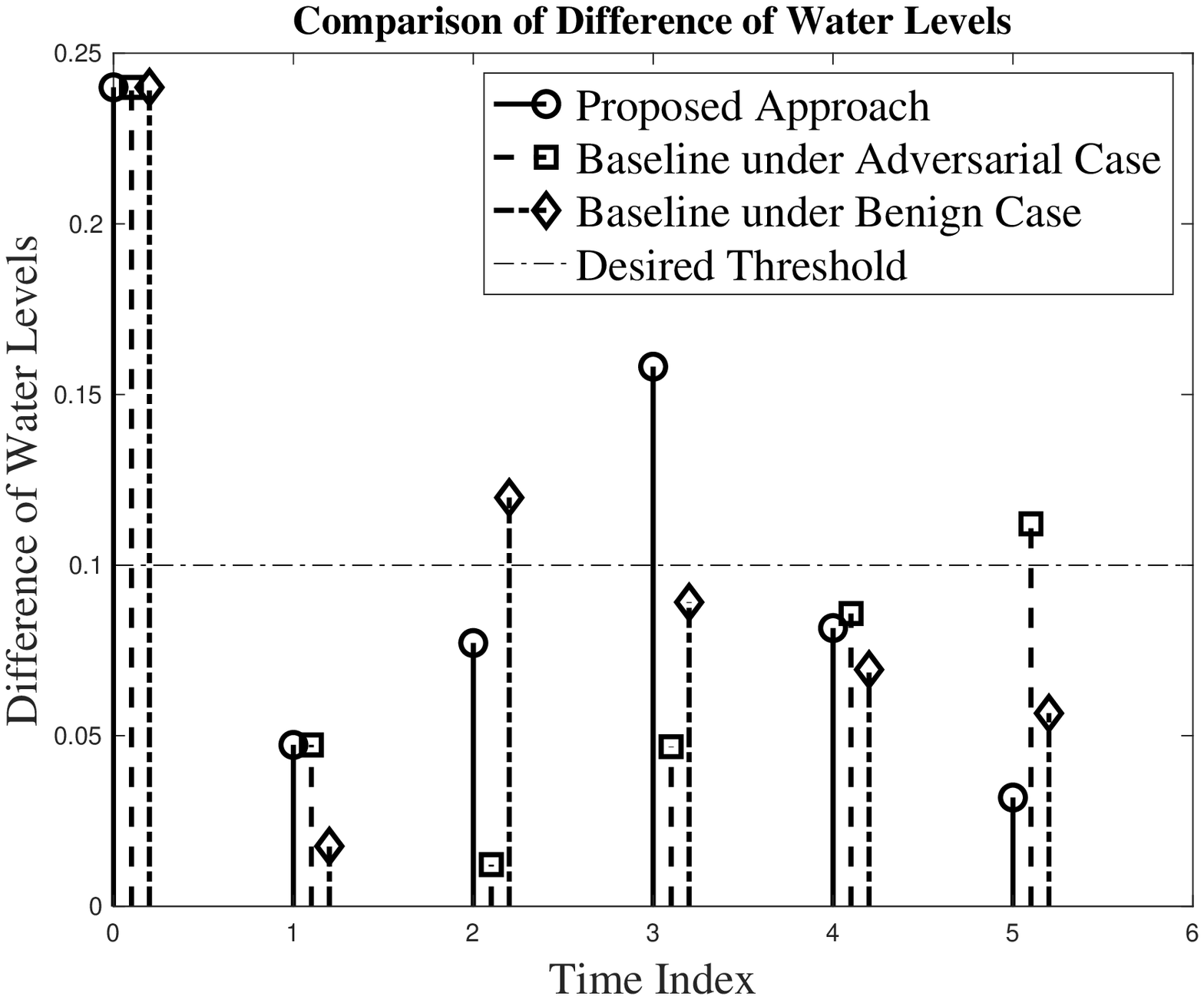}
			\subcaption{}
			\label{fig:difference}
		\end{subfigure}\hfill
		\caption{Evaluation on a two-tank system for an MITL specification that requires water levels in the tanks to be at least $0.3$, and to be within $0.1$ of each other, before time $k=5$. An FSC-based defender policy is compared with a baseline policy that does not account for the presence of an adversary. Fig. \ref{fig:tank2} shows the water level in the second tank, using the two policies. The solid line represents the FSC-based policy, while the dashed and dash-dot lines represent the baseline in the presence and absence of the adversary, respectively. The absolute value of the difference between water levels in the two tanks using the two policies is presented in Fig. \ref{fig:difference}. The solid line with circle markers represents the FSC-based policy, while the dashed line and dash-dot line represent the baseline policy under adversarial and benign environments, respectively. We observe that the baseline policy satisfies the MITL objective in the absence of the adversary, but fails to do so when an adversary is present. The FSC-based policy, in contrast, satisfies the MITL objective in the presence of the adversary.}
	\end{figure*}

	\bibliographystyle{IEEEtran}
	\bibliography{IEEEabrv,MyBib}
	\input{appendix.tex}
\end{document}

%% file: intro.tex
\section{Introduction}
%
%
%
%
%
%
%
%

Cyber-physical systems (CPSs) rely on the smooth integration of physical system components, communication channels, computers, and algorithms \cite{baheti2011cyber}. 
The tight coupling of cyber and physical components introduces additional attack surfaces that can be exploited by an intelligent adversary \cite{banerjee2012ensuring}. In applications such as robotics, the CPS is expected to operate with a large degree of autonomy in dynamic and potentially hazardous environments. Several instances of attacks on CPSs have been recorded and reported, including in vehicles \cite{shoukry2013non}, power systems \cite{sullivan2017cyber}, and nuclear reactors \cite{farwell2011stuxnet}.

Temporal logic (TL) frameworks like linear temporal logic (LTL) enable the expression of system properties such as safety, liveness, and priority \cite{kress2009temporal, ding2014optimal}. 
Off-the-shelf model checking tools can be used to determine if a TL specification can be satisfied by constructing an appropriate finite state automaton (FSA), and searching for a `feasible' path in this FSA \cite{cimatti1999nusmv, kwiatkowska2011prism}. 
The FSA is constructed such that a path in it is `feasible' if and only if the LTL formula is satisfied \cite{baier2008principles}. 
However, a drawback of LTL is that it does not allow for the specification of \emph{time-critical} properties that involve deadlines or intervals. An example of such a property is \emph{`visit a target state every $2$ time units.'} 

We focus on the satisfaction of objectives specified in Metric Interval Temporal Logic (MITL) \cite{alur1996benefits}. MITL uses intervals of length larger than zero to augment timing constraints to modalities of LTL. 
An MITL formula can be represented by a timed automaton (TA) \cite{alur1994theory}. 
TAs extend FSAs by incorporating finitely many clock variables to model the passage of time between successive events. 
Transitions between pairs of states in the TA will then depend on the satisfaction of `clock constraints' in those states. 

We assume that the CPS has to satisfy the MITL objective in the presence of an adversary. 
This could lead to a situation when an adversary could alter the \emph{timed behavior} of the system, thereby causing violation of the objective. 
The adversary is assumed to have the ability to launch attacks on the clocks of the system (\emph{timing attack}) or tamper with inputs to the system (\emph{actuator attack}). A timing attack will prevent the system from reaching desired states within the specified time interval. 
An actuator attack will allow the adversary to steer the system away from a target set of states.

We model the interaction between the defender and adversary as a stochastic game (SG). 
The goal for the defender is then to maximize the probability of satisfying an MITL objective under any adversarial input, while the adversary attempts to thwart the objective by timing and actuator attacks. 
The main challenge in this setting is incorporating time into the SG model, since the adversary has the ability to affect the perception of the (correct) time index by the CPS. 
One approach could be to extend the SG model in \cite{niu2018secure,niu2019optimal} to include time as an additional state, but this will allow the adversary to (unrealistically) effect arbitrary changes of this state by manipulating the timing signal. 
Instead, we define a new type of game that we call durational stochastic games (DSGs) to capture the effect of a timing attack in a principled manner.
%
A DSG probabilistically captures transitions between states, and also the time taken for these transitions. 
It also generalizes SGs (which do not have a notion of time) and semi-Markov decision processes \cite{jewell1963markov} (which have a notion of timed transitions between states, but only assume a single agent giving inputs). 



The defender could incorrectly perceive the time index that it observes if it is the target of a timing attack. 
This incomplete information for the defender makes it computationally challenging to synthesize an optimal policy. 
To address this, we propose the use of finite state controllers (FSCs) to represent the defender policy \cite{hansen2003synthesis}. 
FSCs have been used as policies when the agent is in a partially observable environment. 
An FSC can be viewed as a probabilistic FSA driven by observations of the environment, and producing a distribution over actions of the agent as its output. 
By representing the policy in this way, the defender can commit to policies with finite memory to maintain an estimate of the time index. 
%
This will allow it to synthesize a policy to satisfy the MITL objective even when subject to a timing attack. 

This paper makes the following contributions.
\begin{itemize}
	\item We define a new entity called a \emph{durational stochastic game (DSG)} that captures both time-sensitive objectives and the presence of an adversary.
	%
	\item We construct the defender policy using \emph{finite state controllers (FSCs)}. This will allow it to satisfy the MITL objective in cases of timing attacks. The states of the FSC correspond to the difference between values of the estimated and observed time indices.
	\item We prove that satisfying the MITL formula is equivalent to reaching a subset of states of a \emph{global DSG} constructed by composing representations of the MITL objective, CPS under attack, and FSC. 
	%
	We give a computational procedure to determine this set.
	\item
	%
	We develop a value-iteration based algorithm that maximizes the probability of satisfying the MITL formula for FSCs of fixed sizes under any adversary policy. 
	\item We evaluate our approach on a representation of a signalized traffic network. The adversary is assumed to have the ability to mount actuator and timing attacks on the traffic signals. Our numerical results indicate a significant improvement in the probability of satisfying the given MITL specification compared to two baselines.
\end{itemize}
%

The remainder of this paper is organized as follows. 
Section \ref{Sec:Prelim} gives background on MITL. 
We define the DSG and formally state the problem of interest in Section \ref{Sec:Formulation}. 
Section \ref{Sec:SolutionApproach} presents our main results. 
A numerical case-study is presented in Section \ref{Sec:CaseStudy}. Section \ref{Sec:RelatedWork} discusses related work, and Section \ref{Sec:Conclusion} concludes the paper. The Appendix gives an example on two-tank system and some proofs.

%% file: prelim.tex
\section{Preliminaries}\label{Sec:Prelim}

We introduce MITL and the representation of an MITL formula by a timed automaton. Throughout this paper, we denote by $\mathbb{R}$ the set of real numbers, and by $\mathbb{R}_{\geq0}$ the set of non-negative real numbers. The set of rationals is denoted by $\mathbb{Q}$. The comparison between vectors is component-wise. Bold symbols represent vectors. If $\mathbf{V}\in\mathbb{R}^n$ is a vector of dimension $n$, then $\mathbf{V}(i)$ denotes the $i$-th element of $\mathbf{V}$.
%

Metric Temporal Logic (MTL) \cite{koymans1990specifying} augments timing constraints to the modalities of linear temporal logic (LTL). 
An MTL formula is developed from the same set of atomic propositions $\Pi$ as in LTL and a \emph{time-constrained until} operator $\mathcal{U}_I$, and can be inductively written as: 
	$\varphi:=\top|\pi| \neg \varphi | \varphi_1\wedge \varphi_2|\varphi_1\mathcal{U}_I \varphi_2$, 
%
where $I \subseteq [0, \infty)$ is an interval with endpoints in $\mathbb{N} \cup \{\infty\}$. 
We will focus on \emph{Metric Interval Temporal Logic (MITL)} \cite{alur1996benefits}, a restriction of MTL to intervals $I = [a,b]$ with $a < b$ which is known to be decidable  \cite{alur1996benefits}. 
%

Further, we work with the \emph{point-based semantics}, where MITL formulas are interpreted on \emph{timed words} over an alphabet $2^{\Pi}$. 
A timed word is a sequence $\eta = \{(\pi_i,t_i)\}_{i=0}^{\infty}$ 
where $\pi_i \in 2^{\Pi}$, $t_i \in \mathbb{R}_{\geq 0}$. A time sequence $\{t_i\}_{i=0}^\infty$ associated with any timed word $\eta$ must satisfy the following:
\begin{itemize}
	\item Monotonicity: 
	for all $i \geq 0$, $t_{i+1} > t_i$;
	\item Progress: for all $t\in\mathbb{R}_{\geq 0}$, there exists some $t_i\geq t$.
\end{itemize}
%
\begin{definition}[MITL Semantics]\label{def:MITLSemantics}
	The satisfaction of an MTL formula $\varphi$ at time $t$ by a timed word $\eta$, written $(\eta, t) \models \varphi$, can be recursively defined in the following way: 
	\begin{enumerate}
		\item $(\eta, t) \models \top$ if and only if (iff) $(\eta, 0)$ is true; 
		\item $(\eta, t) \models \pi$ iff $(\eta, t)$ satisfies $\pi$ at time $t$; 
		\item $(\eta, t) \models \neg \varphi$ iff $(\eta, t) \not \models \varphi$; 
		\item $(\eta, t) \models \varphi_1 \wedge \varphi_2$ iff $(\eta, t) \models \varphi_1$ and $(\eta, t) \models \varphi_2$; 
		\item $(\eta, t) \models \varphi_1 \mathcal{U}_I \varphi_2$ iff $\exists k \in I$ such that $(\eta, t+k) \models \varphi_2$ and for all $m <k, (\eta, t+m) \models \varphi_1$. 
	\end{enumerate}
\end{definition}

%
%
MITL admits derived operators in the same way as LTL: 
\emph{i)}: $\varphi_1 \vee \varphi_2\coloneqq\neg(\neg \varphi_1 \wedge \neg \varphi_2)$; 
\emph{ii)}: $\varphi_1 \Rightarrow \varphi_2\coloneqq \neg \varphi_1 \vee \varphi_2$; 
\emph{iii)}: $\Diamond_I\varphi\coloneqq\top \mathcal{U}_I \varphi  \text{  (constrained eventually)}$; 
\emph{iv)}: $\Box_I \varphi:= \neg \Diamond_I \neg \varphi \text{  (constrained always)}$. 
MITL 
and MTL 
also allow for the composition of operators, thereby providing a richer set of specifications. 
For example, $\Diamond_{I_1}\Box_{I_2} \varphi$ means that $\varphi$ will be true at some time within interval $I_1$, and from that time, it will continue to hold for the duration of $I_2$. 
%
%

Given an MITL formula $\varphi$, a timed B\"uchi automaton (TBA) can be constructed to represent $\varphi$ \cite{alur1994theory}. 
In order to do this, we first define a set of clock constraints $\Phi(C)$ over a clock set $C$ as: 
$\phi=\top|\bot|c\bowtie \delta|\phi_1\land\phi_2$,
where $\bowtie\in\{\leq,\geq,<,>\}$, $c\in C$ is a clock, and $\delta\in\mathbb{Q}$ is a non-negative constant. 
A TBA is defined in the following way.
\begin{definition}[Timed B\"{u}chi Automaton \cite{alur1994theory}]\label{def:TBA}
	A timed B\"uchi automaton is a tuple $\mathcal{A} = (Q,2^\Pi,q_0,C,\Phi(C),E,F)$. $Q$ is a finite set of states, $2^\Pi$ is an alphabet over atomic propositions in $\Pi$, $q_0$ is the initial state, $E\subseteq Q\times Q\times2^\Pi\times 2^C\times\Phi(C)$ is the set of transitions, and $F\subseteq Q$ is the set of accepting states. A transition $\langle q,q',a,C',\phi \rangle \in E$ 
	if $\mathcal{A}$ enables the transition from $q$ to $q'$ when a subset of atomic propositions $a\in 2^\Pi$ and clock constraints $\phi\in \Phi(C)$ evaluate to true. The clocks in $C'\subseteq C$ are reset to zero after the transition.
\end{definition}
Given a set of clocks $C$ and $V\subseteq \mathbb{R}^{|C|}$, $\mathbf{v}:C\mapsto V$ is the \emph{valuation} of $C$. 
$\mathbf{v}(c)$ denotes the valuation of a clock $c\in C$. 
The valuation vector $\mathbf{v}$ is then $\mathbf{v}=\left[\mathbf{v}(1),\cdots,\mathbf{v}(|C|)\right]^T$. 
For some $\delta\in\mathbb{Q}$, we define $\mathbf{v}+\delta\coloneqq[\mathbf{v}(1)+\delta,\cdots,\mathbf{v}(|C|)+\delta]^T$.

The \emph{configuration} of $\mathcal{A}$ is a pair $(q,\mathbf{v})$, with $q\in Q$ 
and $\mathbf{v}$ is the valuation defined above. 
A transition $<q,q',a,C',\phi>$ taken after $\delta$ time units from $(q,\mathbf{v})$ to a configuration 
$(q',\mathbf{v}+\delta)$ 
is written $(q,\mathbf{v})\xrightarrow{a,\delta}(q',\mathbf{v}')$, where $\mathbf{v}+\delta\models\phi$ and $\mathbf{v}'(c)=\mathbf{v}(c)+\delta$ for all $c\notin C'$. 
Given an input sequence $a_0,a_1,\cdots$ with $a_i\in\Pi$, we can construct a corresponding sequence of configurations $\rho=(q_0,\mathbf{v}_0)\xrightarrow{a_0,\delta_0}(q_1,\mathbf{v}_1)\cdots$, called a \emph{run} of $\mathcal{A}$. 
The run $\rho$ is \emph{feasible} if for all $i\geq 0$ there exists a transition $<q_i,q_{i+1},a,C_i,\phi>$ in $\mathcal{A}$ such that (i) $\mathbf{v}_0=\mathbf{0}$, (ii) $\mathbf{0}+\delta_0\models\phi_0$, (iii) $\mathbf{v}_1(c)=\mathbf{v}_0(c)+\delta_0$ for all $c\notin C_0$, and (iv) $\mathbf{v}_i+\delta_i\models\phi_i$ and $\mathbf{v}_{i+1}(c)=\mathbf{v}_i(c)+\delta_i$ for all $c\notin C_i$. 
%
%
A feasible run $\rho$ on the TBA $\mathcal{A}$ is \emph{accepting} if and only if it intersects with $F$ infinitely often.

%% file: formulation.tex
\section{Problem Formulation}\label{Sec:Formulation}

In this section, we introduce the adversary and defender models that we will consider in this paper.
We then present an entity called a durational stochastic game (DSG) that models the interaction between the defender and adversary.
The DSG also models the possible amount of time taken for a transition between two states to be completed.
We end the section by formally stating the problem that this paper seeks to solve.

We consider a CPS whose dynamics is given as
\begin{align}
	x(k+1) = f(x(k),u_C(k),u_A(k),w(k)),\label{eq:plant dynamic}
\end{align}
where $k$ is the time index, $x(k)$ is the state of the system, $u_C(k)$ and $u_A(k)$ are the defender's and adversary's inputs, and $w(k)$ is a stochastic disturbance. The time index starts at $k=0$, and is known to both players. The initial state $x(0)$ and statistical information of $w(k)$ is also known to both players. 
The defender aims to synthesize a sequence of inputs to maximize the probability of the MITL objective $\varphi$ being satisfied.
The adversary aims to reduce this probability.
%
\subsection{Adversary and Defender Models}\label{sec:adv model}
%
The adversary can launch an \emph{actuator attack} or a \emph{timing attack}, or a combination of the two to achieve its objective.

During an actuator attack, the adversary manipulates control signals received by the actuator. The sequence of inputs supplied by the adversary in this case is called the \emph{actuator attack policy}, denoted $\tau$.
This attack can be effected when the defender communicates with the actuator via an unreliable communication channel.
In Equation (\ref{eq:plant dynamic}), the adversary can tamper with the control input $u_C(k)$ by injecting a signal $u_A(k)$. Then, the transition of the system to the next state will be
jointly determined by 
$u_C(k)$ and $u_A(k)$.

To launch a timing attack, an adversary can target the time synchronization protocol of the defender \cite{Wang2017Detecting,pasqualetti2013attack}. This will affect the defender's perception of the (correct) time index.
The sequence of inputs supplied by the adversary in this case is called the \emph{timing attack policy}, denoted $\xi$.
The adversary manipulates time stamps $k$ associated with measurements made by the defender as $k+\kappa$, where $\kappa\geq -k$ is an integer.
The policies $\tau$ and $\xi$ will be defined in Section \ref{Problem}.

At each time $k$, the adversary can observe the state $x(k)$ and the correct time index $k$. The observation made by the adversary at time $k$ is defined as $Obs_A^k:=\{x(k),k\}$. The adversary also knows the policy (sequence of inputs) $\mu$ committed to by the defender. Thus, the overall information $\mathcal{I}_A$ available to the adversary is $\mathcal{I}_A:=\bigcup \limits_{m=0:k}Obs_A^{m}\cup \{\mu\}$.
%
%

%
Different from the information available to the adversary, the defender observes the system state $x$ and a time $k'$,
i.e., $Obs_C^{k'}:=\{x(k'),k'\}$, where $x(k')=x(k)$ is the state measurement at time $k$ with possibly incorrect time stamp $k'$ due to a timing attack by the adversary. The overall information available to the defender is $\mathcal{I}_C:=\bigcup \limits_{m=0:k}Obs_C^{m}$. A formal representation of $\mu$ will be given in Section \ref{sec:policy representation}.

\subsection{Durational Stochastic Game}

We present an abstraction of the CPS \eqref{eq:plant dynamic}, that we term a durational stochastic game (DSG). A DSG models the interaction between the defender and adversary, and captures the time taken for a state transition. Let $\Delta$ be a discrete set of possible amounts of time taken for a transition between two states in the DSG, given specific agent actions. Then,
\begin{definition}[Durational Stochastic Game]\label{def:DSG}
	A (labeled) durational stochastic game (DSG) is a tuple $\mathcal{G}= (S_\mathcal{G},s_{\mathcal{G},0},U_C,U_A,Inf_{\mathcal{G},C},Inf_{\mathcal{G},A},Pr_\mathcal{G},T_\mathcal{G},\Pi,L, C)$. $S_\mathcal{G}$ is a finite set of states, $s_{\mathcal{G},0}$ is the initial state. $U_C$, $U_A$ are finite sets of actions and $Inf_{\mathcal{G},C}$, $Inf_{\mathcal{G},A}$ are the information sets of the defender and adversary respectively.
	$Pr_\mathcal{G}: S_\mathcal{G}\times U_C\times U_A\times S_\mathcal{G}\mapsto [0,1]$ encodes $Pr_\mathcal{G}(s'_\mathcal{G}|s_\mathcal{G}, u_C, u_A)$, the transition probability from state $s_\mathcal{G}$ to $s^\prime_\mathcal{G}$ when the controller and adversary take actions $u_C$ and $u_A$.
	$T_\mathcal{G}: S_\mathcal{G}\times U_C\times U_A\times S_\mathcal{G}\times\Delta\mapsto [0,1]$ is a probability mass function. $T_\mathcal{G}(\delta|s_\mathcal{G},u_C,u_A,s'_\mathcal{G})$ denotes the probability that a transition from $s_\mathcal{G}$ to $s'_\mathcal{G}$ under actions $u_C$ and $u_A$ takes $\delta \in \Delta$ time units.
	$\Pi$ is a set of atomic propositions.
	$L:S_\mathcal{G}\mapsto 2^{\Pi}$ is a labeling function that maps each state to atomic propositions in $\Pi$ that are true in that state, and $C$ is the set of clocks.
	%
\end{definition}

In this work, we assume the transition probability $Pr_\mathcal{G}$ and probability mass function $T_\mathcal{G}$ are known to both the defender and adversary. In Definition \ref{def:DSG}, the transition probability between states is jointly determined by actions taken by the defender and adversary, which 
models an actuator attack. The asymmetry of information sets of the two agents
models a timing attack. This can be justified as follows:
%
%
let the actions available to the agents at a state $s\in S_\mathcal{G}$ be $U_C(s)$ and $U_A(s)$, and let the respective information sets be $Inf_{\mathcal{G},C}(s)$ and $Inf_{\mathcal{G},A}(s)$. In order to capture the information pattern described in Section \ref{sec:adv model},
we have $Inf_{\mathcal{G},C}(s)=\{(s_0,\mathbf{v}_0),\cdots,(s,\Bar{\mathbf{v}})\}$, i.e., the defender knows the path from the initial state $s_0$ to current state $s$ along with the time stamp of each state being reached. We reiterate that the time stamps observed by the defender could have been manipulated by the adversary, and hence may be incorrect. The adversary knows the path from the initial state $s_0$ to current state $s$ along with the correct time stamps of each state being reached, and the defender policy, i.e., $Inf_{\mathcal{G},A}(s)=\{(s_0,\mathbf{v}_0),\cdots,(s,\mathbf{v})\}\cup\{\mu\}$.

For the remainder of this paper, we use the DSG $\mathcal{G}$ as an abstraction of the CPS described in Section \ref{sec:adv model}. 
The mapping from the CPS model \eqref{eq:plant dynamic} to a DSG is presented in Algorithm \ref{alg:abstract}. Algorithm \ref{alg:abstract} partitions the state space and the admissible control and adversary action sets (lines 5-6). We use Monte-Carlo simulation \cite{cappe2007overview} to compute the transition probability distributions $Pr_\mathcal{G}$ and and $T_\mathcal{G}$ (lines 8-17).
%
%
%
\subsection{Problem Statement} \label{Problem}

%
Comparing the information sets of the two agents, we observe that
the adversary receives more information than the defender, including the correct time and the defender's policy. This asymmetric information pattern can be modeled as a Stackelberg game \cite{fudenberg1991game}, with the defender as leader and the adversary as follower.
In this paper, we use finite state controllers to represent the policy of the defender.
For the time-being, however, it will suffice to think of the defender's policy as a probability distribution over the defender actions, given the state of DSG. The adversary policies corresponding to the two types of attacks is formally stated below.

\begin{definition}[Adversary policies] \label{def: AdvPolicies}
The \emph{actuator attack policy} is a map $\tau:S_\mathcal{G}\times V\mapsto U_A$. That is, $\tau$ specifies an action $u_A\in U_A(s)$ for each state $(s,\mathbf{v})\in S$.

The \emph{timing attack policy} is a map $\xi:V\times V\mapsto[0,1]$. That is, $\xi$ encodes $\xi(\mathbf{v}'|\mathbf{v})$, the probability that the adversary will manipulate the correct clock valuation $\mathbf{v}$ to a valuation $\mathbf{v}'$.
\end{definition}

We define a Stackelberg equilibrium, which indicates that a solution to a Stackelberg game has been found.
Denote the leader's policy by $\mu$ and follower's policy by the tuple $(\tau, \xi)$.
Let $Q_L(\mu, (\tau, \xi))$ and $Q_F(\mu, (\tau, \xi))$ be the utilities gained by the leader and follower by adopting their respective policies.
\begin{definition}[Stackelberg Equilibrium (SE)]\label{StackEq}
	A tuple $(\mu, (\tau, \xi))$ is an \emph{SE} if $\mu = \arg \max \limits_{\mu'}Q_L(\mu', BR(\mu'))$, where $BR(\mu') = \{(\tau, \xi): (\tau, \xi)=\arg \max Q_F(\mu',(\tau, \xi))\}$.
	That is, the leader's policy is optimal given that the follower observes this and plays its best response.
\end{definition}
%
%
We are now ready to state the problem.
\begin{problem}\label{prob:formulation}
	Given an MITL objective $\varphi$, and a DSG in which the defender's objective is to maximize the probability of satisfying $\varphi$ and the adversary's objective is to minimize this probability, compute a control policy that is in SE, i.e.,
	\begin{equation}
	\max_\mu\min_{\tau,\xi}\mathbb{P}(\varphi).
	\end{equation}
\end{problem}

\begin{algorithm}
	\caption{Constructing a DSG abstraction for CPS.}
	\label{alg:abstract}
	\begin{algorithmic}[1]
		\Procedure{Construct\_DSG}{}
		\State \textbf{Input:} CPS model $f(x(k),u_C(k),u_A(k),w(k))$
		\State \textbf{Output:} DSG $\mathcal{G}$
		\State Initialize time-horizon $K$
		\State Partition the state space as $\mathcal{X}=\cup_{i=1}^nX_i$
		\State Partition control and adversary input as sets of polytopes $U_C=\{u_{C_1},\cdots,u_{C_\Xi}\}$, $U_A=\{u_{A_1},\cdots,u_{A_\Gamma}\}$
		\State $S=\{X_1,\ldots,X_n\}$ and $\mathcal{L}$ is determined accordingly
		\For{$l=1,\ldots,n$}
		\For{all $u_{C} \in U_{C}$ and $u_{A} \in U_{A}$}
		\For{$k=1,\ldots,K$}
		\State $x \leftarrow$ sampled state in $X_{i}$
		\State $\hat{u}_{C}, \hat{u}_{A} \leftarrow$ sampled  inputs from $u_{C},u_{A}$
		\State $j \leftarrow$ region containing $f(x, \hat{u}_{C}, \hat{u}_{A},\vartheta)$
		\State Use particle filter to approximate transition probabilities $Pr_\mathcal{G}$ and duration function $T_\mathcal{G}$ between sub-regions $X_i$ and $X_j$ for all $i$ and $j$.
		\EndFor
		\EndFor
		\EndFor
		\EndProcedure
	\end{algorithmic}
\end{algorithm}

%% file: sol.tex
\section{Solution Approach}\label{Sec:SolutionApproach}

This section presents the main results of the paper. 
We first compute a product durational stochastic game (PDSG), given a DSG that abstracts the CPS, and a TBA corresponding to the MITL formula $\varphi$. 
We represent the defender's policy as a finite state controller (FSC), and compute a global DSG (GDSG) by composing the PDSG and FSC.
We solve Problem \ref{prob:formulation} by proving that maximizing the probability of satisfying $\varphi$ is equivalent to maximizing the probability of reaching a subset of states of the GDSG, termed \emph{generalized accepting end components (GAMECs)}. 
Then, we present a value-iteration based algorithm to synthesize an FSC that will lead to an SE of the game between defender and adversary. 
%
\subsection{Product Durational Stochastic Game Construction}
%
\begin{definition}[Product Durational Stochastic Game]\label{def:PDSG}
	A PDSG $\mathcal{P}$ constructed from a DSG $\mathcal{G}$ 
	, TBA $\mathcal{A}$, and clock valuation set $V$ 
	is a tuple $\mathcal{P}=(S,s_0,U_C,U_A,Inf_C,Inf_A,Pr,Acc)$. The set $S=S_\mathcal{G}\times Q\times V$ is a finite set of states, $s_0=(s_{\mathcal{G},0},q_0,\mathbf{v}_0)$ is the initial state, $U_C$, $U_A$ are finite sets of actions and $Inf_C$, $Inf_A$ are the information sets of the defender and adversary respectively. $Pr: S\times U_C\times U_A\times S\mapsto [0,1]$ encodes $Pr\left((s^{\prime},q^\prime,\mathbf{v}')|(s,q,\mathbf{v}), u_C, u_A\right)$, the probability of a transition from state $(s,q,\mathbf{v})$ to $(s',q',\mathbf{v}')$ when the defender  and adversary take actions $u_C$ and $u_A$ respectively. The probability 
	\begin{multline}\label{eq:transition prob}
		Pr\left((s^{\prime},q^\prime,\mathbf{v}')|(s,q,\mathbf{v}), u_C, u_A\right)\\\coloneqq T_\mathcal{G}(\delta|s,u_C,u_A,s')Pr_\mathcal{G}(s'|s,u_C,u_A)
	\end{multline}
	if and only if $(q,\mathbf{v})\xrightarrow{L(s'),\delta}(q',\mathbf{v}')$. 
	$Acc=S_\mathcal{G}\times F \times V$ is a finite set of accepting states. 
\end{definition}


At a state $(s,q,\mathbf{v})\in S$, 
let $Inf_C(s,q,\mathbf{v}):=\{(s_0,q_0,\mathbf{v}_0),\cdots,(s,q,\bar{\mathbf{v}})\}$ (the defender knows the path from the initial state of PDSG to the current state, along with the manipulated time stamps) 
and $Inf_A(s,q,\mathbf{v}):=\{(s_0,q_0,\mathbf{v}_0),\cdots,(s,q,\mathbf{v})\}\cup\{\mu\}$ (the adversary knows the defender's policy $\mu$ and the path from the initial state to the current state, along with the correct time stamps).

The following result establishes the consistency of the PDSG $\mathcal{P}$. The proof can be found in the Appendix.
\begin{proposition} \label{Prop: Consistency}
	The function $Pr(\cdot)$ is well-defined. 
	That is, $Pr\left((s^{\prime},q^\prime,\mathbf{v}')|(s,q,\mathbf{v}), u_C, u_A\right)\in[0,1]$ and
	\begin{equation}\label{eq:well defined prob}
	\sum_{(s',q',\mathbf{v}')}Pr\left((s^{\prime},q^\prime,\mathbf{v}')|(s,q,\mathbf{v}), u_C, u_A\right)=1.
	\end{equation} 
\end{proposition}

From \eqref{eq:transition prob}, we observe that a transition exists in $\mathcal{P}$ if and only if the label associated with the target state matches the atomic proposition corresponding to the transition in the TBA, and the clock constraint is satisfied. 
Further, for any 
run $\beta : = (s_0,q_0,\mathbf{v}_0), (s_1,q_1,\mathbf{v}_1),\dots$ on $\mathcal{P}$, we can obtain a run $\rho$ on $\mathcal{A}$ and a path on $\mathcal{G}$.
That is, there is a one-one mapping from runs on the PDSG to those on the TBA and DSG. 
We define the following two projections over the runs on $\mathcal{P}$. Given a run $\beta$, we let $\mathsf{Untime}(\beta)=(s_0,q_0),(s_1,q_1),\cdots,$ be the untimed sequence of states, and let $\mathsf{Time}(\beta)=(q_0,\mathbf{v}_0),(q_1,\mathbf{v}_1),\cdots,$ be the configuration sequence corresponding to $\beta$.

\subsection{Defender Policy Representation: Finite State Controllers}\label{sec:policy representation}

We now formally define the defender's policy $\mu$. 
Since the adversary can manipulate the clock valuation $\mathbf{v}$ observed by the defender, the defender has only partial information over the DSG. 
This is evident from the following: let there exist a run $\beta=(s_0,q_0,\mathbf{0})(s_1,q_1,\mathbf{1})(s_2,q_2,\mathbf{2})$ on PDSG $\mathcal{P}$ without any clock being reset that is manipulated by the adversary as $\beta'=(s_0,q_0,\mathbf{0})(s_1,q_1,\mathbf{1})(s_2,q_2,\mathbf{0.5})$. The run $\beta'$ is not reasonable since the time sequence $\mathsf{Time}(\beta')=(q_0,\mathbf{0}),(q_1,\mathbf{1}),(q_2,\mathbf{0.5})$ is not monotone. 
The presence of such a run will allow the defender to conclude that a timing attack has been effected by the adversary. 
Moreover, after a timing attack has been detected, the defender will be aware that the observed clock valuation is incorrect, and thus cannot be relied upon for control synthesis. 
The defender will then need to keep track of an estimate of the clock valuation in order to detect a timing attack, and use this estimate for control synthesis. 
The defender's policy is represented as a finite state controller (FSC) defined as follows.

\begin{definition}[Finite State Controller \cite{hansen1998solving}]\label{def:policy graph}
	A finite state controller (FSC) is a finite state automaton $\mathcal{F}=(Y,y_0,\mu)$, where $Y=\Lambda\times\{0,1\}$ is a finite set of internal states, $\Lambda$ is a set of estimates of clock valuations, the set $\{0,1\}$ indicates if a timing attack has been detected ($1$) or not ($0$). $y_0$ is the initial internal state. $\mu$ is the defender policy, given by: 
	\begin{equation}\label{defpolicy}
	\mu=\begin{cases}
	\mu_0:Y\times S\times Y\times U_C\mapsto[0,1],\mbox{ if } \mathcal{H}_0 \text{ holds};\\
	\mu_1:Y\times S_\mathcal{G}\times Q\times Y\times U_C\mapsto [0,1],\mbox{ if } \mathcal{H}_1\text{ holds},
	\end{cases}
	\end{equation}
	where $\mu_0$ and $\mu_1$ respectively denote the control policies that will be executed when hypothesis $\mathcal{H}_0$ or $\mathcal{H}_1$ holds.
\end{definition}

In Definition \ref{def:policy graph}, the hypothesis $\mathcal{H}_0$ models the scenario where no timing attack has been detected by the defender, and $\mathcal{H}_1$ models the case when a timing attack has been detected. Equation (\ref{defpolicy}) specifies the probability of reaching the next internal state $y'$ and taking the corresponding action $u_C$, given the current internal state $y$, observed clock valuation $\mathbf{v}$ (if no timing attack has been detected), and state $s$ of $\mathcal{G}$. 
The FSC allows the defender to synthesize policies with finite memory rather than memoryless policies. In this paper, we assume the size of the FSC is given and fixed 
and limit our focus to computing $\mu$. The two players can track an estimate of the clock valuation according to the probability distribution $T_\mathcal{G}$. 
Moreover, we do not explicitly specify a timing attack detection scheme, and assume it is known. Timing attack detection schemes that are compatible with our framework include \cite{Wang2017Detecting} and \cite{pasqualetti2013attack}. 
In the nominal case, 
the defender adopts the policy $\mu_0$. Once a timing attack has been detected by the defender, the defender ignores the observed clock valuation $\mathbf{v}$, and switches to policy $\mu_1$. 
The design of a timing attack detection strategy is beyond the scope of this paper, and we leave it as future work.

%
%

\subsection{Proposed Solution}
To incorporate the evolution of the estimate of the clock valuation maintained by the defender, we compose this with the PDSG. We call this entity the global DSG (GDSG). 
We prove that maximizing the probability of satisfying the MITL objective $\varphi$ is equivalent to maximizing the probability of reaching a specific subset of states in the GDSG called generalized accepting maximal end components (GAMECs). The control policy is then computed using a value iteration based procedure.
Given an FSC $\mathcal{F}$ and the PDSG $\mathcal{P}$, we can construct GDSG in the following way.
\begin{definition}[Global DSG (GDSG)]\label{def:GDSG}
	A GDSG is a tuple $\mathcal{Z}=(S_\mathcal{Z},s_{\mathcal{Z},0},U_C,U_A,Inf_{\mathcal{Z},C}, Inf_{\mathcal{Z},A},Pr_\mathcal{Z},Acc_\mathcal{Z})$, where $S_\mathcal{Z}=S\times Y$ is a finite set of states, $s_{\mathcal{Z},0}=(s_0,q_0,\mathbf{v}_0,y_0)$ is the initial state. $U_C$ and $U_A$ are finite sets of actions and $Inf_{\mathcal{Z},C}$ and $Inf_{\mathcal{Z},A}$ are the information sets of the defender and adversary respectively. 
	$Pr_\mathcal{Z}: S_\mathcal{Z}\times U_C\times U_A\times S_\mathcal{Z}\mapsto [0,1]$ is a transition function where $Pr_\mathcal{Z}\left((s^{\prime},q^\prime,\mathbf{v}',y')|(s,q,\mathbf{v},y), u_C, u_A\right)$ is the probability of a transition from state $(s,q,\mathbf{v},y)$ to $(s',q',\mathbf{v}',y)$ when the defender and adversary take actions $u_C$ and $u_A$ respectively. The transition probability 
	\begin{multline}\label{eq:GDSG transition prob}
		Pr_\mathcal{Z}\left((s^{\prime},q^\prime,\mathbf{v}',y')|(s,q,\mathbf{v},y), u_C, u_A\right)=\\  
		\begin{cases} \sum_{\mathbf{v}''}\xi(\mathbf{v}''|\mathbf{v})\mu_0(y',u_C|s,q,\mathbf{v}'',y)\\
			\cdot Pr\left((s^{\prime},q^\prime,\mathbf{v}')|(s,q,\mathbf{v}), u_C, u_A\right),\mbox{ if }\mathcal{H}_0\text{ holds};\\
			\mu_1(y',u_C|s,q,y)T_\mathcal{G}(\delta|s,u_C,u_A,s')\\
			\cdot Pr_\mathcal{G}(s'|s,u_C,u_A),\quad\quad\quad\quad\quad\quad\mbox{ if }\mathcal{H}_1\text{ holds};
		\end{cases}
	\end{multline}
	$Acc_\mathcal{Z}=Acc \times Y$ is the set of accepting states. 
\end{definition}

At a state $(s,q,\mathbf{v},y)$, the information set 
of the defender is $Inf_{\mathcal{Z},C}(s,q,\mathbf{v},y)=\{(s_0,q_0,\mathbf{v}_0,y_0),\cdots,(s,q,\bar{\mathbf{v}},y)\}$. That is, the defender knows the path from the initial state of GDSG to the current state, along with the time stamps, which might have been manipulated by the adversary. The information set of the adversary is $Inf_{\mathcal{Z},A}(s,q,\mathbf{v},y)=\{(s_0,q_0,\mathbf{v}_0,y_0),\cdots,(s,q,\mathbf{v},y)\}\cup\{\mu\}$. That is, the adversary knows the path from the initial state to the current state, along with the correct time stamps, and the defender's policy. In the sequel, we focus on the GDSG $\mathcal{Z}$ in Definition \ref{def:GDSG}, and denote a state $(s,q,\mathbf{v},y)$ in $\mathcal{Z}$ as $\mathfrak{s}$. Given a run $\beta=\{(s_i,q_i,\mathbf{v}_i,y_i)\}_{i\geq 1}$ on $\mathcal{Z}$, we define $\mathsf{Untime}(\beta)=\{(s_i,q_i)\}_{i\geq 1}$ and $\mathsf{Time}(\beta)=\{(q_i,\mathbf{v}_i)\}_{i\geq 1}$, respectively. To compute the control policy that satisfies $\varphi$, we need to determine accepting runs on $\mathcal{Z}$. To this end, we introduce the concepts of generalized maximal end component (GMEC) and generalized accepting maximal end component (GAMEC) \cite{niu2019optimal}. We note that the accepting condition for a GMEC in this paper differs from that in \cite{niu2019optimal}, since we are working with timed automata.
\begin{definition}[Sub-DSG]\label{def: subSG}
	A sub-DSG of a DSG $\mathcal{G}= (S,U_C,U_A,Pr,s_0,\Pi,\mathcal{L})$ is a tuple $(N,D)$ where $\emptyset\neq N\subseteq S$ is a set of states, and $D:N\rightarrow 2^{U_C}$ is a function such that $D(s)\subseteq U_C(s)$ for all $s\in N$ and $\{s^\prime|Pr(s'|s,u_C,u_A)>0,\forall u_A\in U_A(s),s\in N\}\subseteq N$.
\end{definition}
%
%
\begin{definition}
\label{def: GMEC}
	A Generalized End Component (GEC) is a sub-DSG $(N,D)$ such that the underlying directed graph $G_{(N,D)}$ of $(N,D)$ is strongly connected. A GMEC is a GEC $(N,D)$ such that there exists no other GEC $(N^\prime,D^\prime)\neq (N,D)$, where $N\subseteq N^\prime$ and $D(s)\subseteq D^\prime(s)$ for all $s\in N$.
	
	A GAMEC is a GMEC if $Acc\cap N\neq\emptyset$.
\end{definition}
%

Algorithm \ref{algo:GAMEC} presents a procedure to compute the set of GAMECs, $\mathcal{C}$ of the GDSG $\mathcal{Z}$. 
Let 
$\mathcal{E}$ denote the set of states in a GAMEC. The correctness of Algorithm \ref{algo:GAMEC} is established in the following result. The proof can be found in the Appendix. 
\begin{proposition}\label{proposition:correctness}
	Algorithm \ref{algo:GAMEC} returns all GAMECs of $\mathcal{Z}$.
\end{proposition}
	\begin{algorithm}[!h]
		\caption{Computing the set of GAMECs $\mathcal{C}$.}
		\label{algo:GAMEC}
		\begin{algorithmic}[1]
			\Procedure{Compute\_GAMEC}{$\mathcal{Z}$}
			\State \textbf{Input}: GDSG $\mathcal{Z}$
			\State \textbf{Output:} Set of GAMECs $\mathcal{C}$
			\State \textbf{Initialization:} Let $D(\mathfrak{s})\leftarrow U_C(\mathfrak{s})$ for all $\mathfrak{s}\in S$. Let $\mathcal{C}\leftarrow \emptyset$ and $\mathcal{C}_{temp}\leftarrow\{S\}$
			\Repeat 
			\State $\mathcal{C}\leftarrow\mathcal{C}_{temp}$, $\mathcal{C}_{temp}\leftarrow\emptyset$
			\For{$N\in\mathcal{C}$}
			\State $R\leftarrow\emptyset$ 
			\State Let $SCC_1,\cdots,SCC_n$ be the set of strongly connected components of underlying digraph $G_{(N,D)}$
			\For{$i=1,\cdots,n$}
			\For{each state $\mathfrak{s}\in SCC_i$}
			\State $D(\mathfrak{s})\leftarrow\{u_C\in U_C(\mathfrak{s})|\mathfrak{s}^\prime\in N,Pr(\mathfrak{s}^\prime|\mathfrak{s},u_C,u_A)>0,~\forall u_A\in U_A(\mathfrak{s})\}$
			\If{$D(s)=\emptyset$}
			\State $R\leftarrow R\cup\{\mathfrak{s}\}$
			\EndIf
			\EndFor
			\EndFor
			\While{$R\neq\emptyset$}
			\State dequeue $\mathfrak{s}\in R$ from $R$ and $N$
			\If{$\exists \mathfrak{s}^\prime\in N$ and $u_C\in U_C(\mathfrak{s}^\prime)$ such that $Pr(\mathfrak{s}|\mathfrak{s}^\prime,u_C,u_A)>0$ for some $u_A\in U_A(\mathfrak{s}^\prime)$}
			\State $D(\mathfrak{s}^\prime)\leftarrow D(\mathfrak{s}^\prime)\setminus\{u_C\}$
			\If{$D(\mathfrak{s}^\prime)=\emptyset$}
			\State $R\leftarrow R\cup\{\mathfrak{s}^\prime\}$
			\EndIf
			\EndIf
			\EndWhile
			
			\For{$i=1,\cdots,n$}
			\If{$N\cap SCC_i\neq\emptyset$}
			\State $\mathcal{C}\leftarrow\mathcal{C}_{temp}\cup\{N\cap SCC_i\}$
			\EndIf
			\EndFor
			\EndFor
			\Until{$\mathcal{C}=\mathcal{C}_{temp}$}
			\For {$N\in\mathcal{C}$}
			\If{$Acc_\mathcal{Z}\cap N=\emptyset$}
			\State $\mathcal{C}=\mathcal{C}\setminus N$
			\EndIf
			\EndFor
			\State \Return $\mathcal{C}$
			\EndProcedure
		\end{algorithmic}
	\end{algorithm}
%
The equivalence between the satisfying the MITL objective $\varphi$ and reaching states in the GAMEC is stated below:
\begin{theorem}\label{Theorem:Equiv}
	Given an initial state $\mathfrak{s}_0\in S_{\mathcal{Z}}$, the minimum probability of satisfying $\varphi$ is equal to the minimum probability of reaching the states $\mathcal{E}$ of GAMEC. That is, 
	\begin{equation}
	\min_{\tau,\xi}\mathbb{P}(\varphi|\mathfrak{s}_0)=\min_{\tau,\xi}\mathbb{P}(\text{reach } \mathcal{E}|\mathfrak{s}_0),
	\end{equation}
	where $\mathbb{P}(\varphi|\mathfrak{s}_0)$ and $\mathbb{P}(\text{reach } \mathcal{E}|\mathfrak{s}_0)$ are the probabilities of satisfying $\varphi$ and reaching $\mathcal{E}$ when starting from $\mathfrak{s}_0$.
\end{theorem}
To prove Theorem \ref{Theorem:Equiv}, we need an intermediate result \cite{alur1994theory}.
\begin{lemma}\label{lemma:untime language}
	Let $L$ be a timed regular language. Then, a word 
	$\{\rho_i\}_{i\geq 1}\in\mathsf{Untime}(L)$ if and only if there exists a sequence $\{t_i\}_{i\geq 1}$ such that $t_i\in\mathbb{Q}$ and the timed word $\{\rho_i,t_i\}_{i\geq 1}\in L$.
\end{lemma} 

Lemma \ref{lemma:untime language} indicates that we can analyze a timed word by focusing on its untimed projection and the corresponding time sequence. We use this to prove Theorem \ref{Theorem:Equiv}.

\begin{proof}[Proof of Theorem \ref{Theorem:Equiv}]
	We establish that satisfying $\varphi$ is equivalent to reaching the set of states $\mathcal{E}$. Then, we need to show that any accepting run will reach $\mathcal{E}$, and any run that reaches $\mathcal{E}$ is accepting. Let $\mathcal{L}$ denote the timed language accepted by $\mathcal{Z}$. Let $\mathsf{Untime}(\mathcal{L})$ be the language obtained from $\mathcal{L}$ by discarding the clock valuation and internal state components. 
	
	First, we prove that any accepting run $\beta$ of GDSG $\mathcal{Z}$ reaches $\mathcal{E}$. We use a contradiction argument. Suppose there exists an accepting run $\beta$ that does not reach $\mathcal{E}$. Since $\beta$ satisfies $\varphi$, we must have that $\beta$ contains some accepting state in $Acc_\mathcal{Z}$ infinitely many times (Definition \ref{def:TBA}). This implies that there exists a GEC that contains some state $\mathfrak{s}\in Acc_\mathcal{Z}$ and $\mathfrak{s}\notin \mathcal{E}$, which violates Proposition \ref{proposition:correctness}.
	
	Next, we show that any run $\beta$ that reaches $\mathcal{E}$ is accepting. We use Lemma \ref{lemma:untime language}. 
	Since GAMECs are strongly connected and each GAMEC contains at least one accepting state, reaching $\mathcal{E}$ is equivalent to reaching some accepting state infinitely often, which agrees with the acceptance condition of $\mathsf{Untime}(\mathcal{L})$. Hence, we have $\mathsf{Untime}(\beta)\in\mathsf{Untime}(\mathcal{L})$. 
	Now, from Equation \eqref{eq:transition prob}, we have that a transition in $\mathcal{P}$, and therefore in $\mathcal{Z}$, exists if and only if no clock constraint is violated, i.e., $(q,\mathbf{v})\xrightarrow{L(s'),\delta}(q',\mathbf{v}')$. Otherwise, the transition probability is $0$, and hence the run $\beta$ does not exist, which establishes the claim. Given that $\mathsf{Untime}(\beta)\in \mathsf{Untime}(L)$ holds, and $\mathsf{Time}(\beta)$ never violates the clock constraints defined by TBA $\mathcal{A}$ for any run that reaches $\mathcal{E}$, we have that the set of runs that reach $\mathcal{E}$ is in language $\mathcal{L}$ by Lemma \ref{lemma:untime language}. 
	
	Combining the two arguments above, we observe that satisfying $\varphi$ is equivalent to reaching the set $\mathcal{E}$. 
	This gives $\min \limits_{\tau,\xi}\mathbb{P}(\varphi|\mathfrak{s}_0)=\min \limits_{\tau,\xi}\mathbb{P}(\text{reach } \mathcal{E}|\mathfrak{s}_0)$, completing the proof.
\end{proof}
%
%

%


Let the vector $\mathbf{Q}(\mathfrak{s})\in\mathbb{R}^{|S_\mathcal{Z}|}$ represent the probability of satisfying $\varphi$ when starting from a state $\mathfrak{s}=(s,q,\mathbf{v},y)$ in $\mathcal{Z}$.
\begin{proposition}\label{proposition:value}
	Let $\mathbf{Q}:= \max \limits_{\mu}{\min \limits_{\tau, \xi}{\mathbb{P}(\varphi)}}$ be the probability of satisfying $\varphi$. 
	Then,
	\begin{multline}\label{eq:Stackelberg-Shapley}
		\mathbf{Q}((s,q,\mathbf{v},y)) = \max_\mu\min_{\tau,\xi}\sum_{u_{C} \in U_{C}}\sum_{u_{A} \in U_{A}}\sum_{(s',q',\mathbf{v}',y) \in S_\mathcal{Z}}\\\tau((s,q,\mathbf{v},y),u_{A})\mathbf{Q}((s',q',\mathbf{v}',y'))\\ Pr_\mathcal{Z}\left((s^{\prime},q^\prime,\mathbf{v}',y')|(s,q,\mathbf{v},y), u_C, u_A\right),~\forall (s,q,\mathbf{v},y).
	\end{multline}
	Moreover, the value vector is unique.
\end{proposition}
Before proving Proposition \ref{proposition:value}, we define the operators:
\begin{align*}{rCl}
	(M_{\mu}\mathbf{Q})(\mathfrak{s}) &= \min_{\tau,\xi}{\sum_{\mathfrak{s}^{\prime}}{P(\mathfrak{s}^{\prime}|\mathfrak{s},\mu,(\tau,\xi))\mathbf{Q}(\mathfrak{s}^{\prime})}},\\
	(M\mathbf{Q})(\mathfrak{s}) &= \max_{\mu}{\min_{\tau,\xi}{\sum_{\mathfrak{s}^{\prime}}{P(\mathfrak{s}^{\prime}|\mathfrak{s},\mu,(\tau,\xi))\mathbf{Q}(\mathfrak{s}^{\prime})}}},
\end{align*}
where $P(\mathfrak{s}^{\prime}|\mathfrak{s},\mu,(\tau,\xi))$ is the probability of transiting from state $\mathfrak{s}$ to $\mathfrak{s}'$, given policies $\mu$ and $\tau$. The operators $M_\mu$ and $M$ are characterized in the following lemma. The proof of the lemma can be found in the Appendix.
\begin{lemma}\label{lemma:convergence}
	The sequence of value vectors obtained by composing operators $M_\mu$ and $M$ is convergent.
\end{lemma}
\begin{proof}[Proof of Proposition \ref{proposition:value}]
	We prove by contradiction. 
	Let $\mathbf{Q}$ be a value vector associated with the control policies in SE but Equation \eqref{eq:Stackelberg-Shapley} does not hold. 
	Let $\mathbf{Q}^*$ be the probability of satisfying $\varphi$ under control policies in SE. 
	Since $\mathbf{Q}$ is the value obtained for a control policy $\mu$ and adversary policies $(\tau,\xi)$ that are the best responses to $\mu$, $\mathbf{Q}=M_\mu \mathbf{Q}\leq M\mathbf{Q}$. 
	Composing $M_\mu$ and $M$ $k$ times and as $k \rightarrow \infty$, we have 
	\begin{equation*}
		\mathbf{Q}=\lim_{k\rightarrow\infty}M^k_\mu\mathbf{Q}\leq \lim_{k\rightarrow\infty}M^k\mathbf{Q}=\mathbf{Q}^*,
	\end{equation*}
	where the first and last equalities hold by Lemma \ref{lemma:convergence}, and the inequality follows from definitions of $M$ and $M_\mu$. 
	If $\mathbf{Q}<\mathbf{Q}^*$, we have that the policy $\mu$ is not in SE. 
	If $\mathbf{Q}=\mathbf{Q}^*$, then Equation \eqref{eq:Stackelberg-Shapley} holds, which contradicts the hypothesis.
	
	Now, suppose there exist value vectors $\mathbf{Q}$ and $\mathbf{Q}'$ such that $\mathbf{Q}\neq\mathbf{Q}'$ and their corresponding control policies $\mu$ and $\mu'$ are both in SE. 
	By assumption, we have $\mathbf{Q}=M\mathbf{Q}\geq M_\mu'\mathbf{Q}$. Composing both sides of the inequality $k$ times and letting $k \rightarrow \infty$, from Lemma \ref{lemma:convergence}, $\mathbf{Q}^*=\lim \limits_{k\rightarrow\infty}M^k\mathbf{Q}\geq \lim \limits_{k\rightarrow\infty}M_{\mu'}^k\mathbf{Q}=\mathbf{Q}'$. 
	If $\mathbf{Q}>\mathbf{Q}'$, we have that policy $\mu'$ is not in SE. 
	Thus, we must have $\mathbf{Q}=\mathbf{Q}'$ so that the policies $\mu$ and $\mu'$ are both in SE, which contradicts our initial assumption.
\end{proof}
These results enable determining an optimal control policy using a value-iteration based algorithm. 
Algorithm \ref{algo:value} computes the value vector at each iteration. 
The value vector is updated following Proposition \ref{proposition:value}. Given the optimal value vector $\mathbf{Q}^*$ and the Stackelberg setting, we can extract the optimal defender's policy as the maximizer of $\mathbf{Q}^*$ by solving a linear program.
The convergence of Algorithm \ref{algo:value} is discussed in Theorem \ref{thm:convergence}. The proof uses an inductive argument, and can be found in the Appendix.
%
\begin{algorithm}
	\caption{Computing an optimal control policy.}
	\label{algo:value}
	\begin{algorithmic}[1]
		\Procedure{Control\_Synthesis}{$\mathcal{Z}$}
		\State \textbf{Input}: GDSG $\mathcal{Z}$
		\State \textbf{Output:} value vector $\mathbf{Q}$
		\State \textbf{Initialization:} $\mathbf{Q}^{0} \leftarrow \mathbf{0}$, $\mathbf{Q}^{1}(\mathfrak{s}) \leftarrow 1$ for $\mathfrak{s} \in Acc_\mathcal{Z}$, $\mathbf{Q}^{1}(\mathfrak{s}) \leftarrow 0$ otherwise, $k \leftarrow 0$
		\While{$\max{\{|\mathbf{Q}^{k+1}(\mathfrak{s})-\mathbf{Q}^{k}(\mathfrak{s})| : \mathfrak{s} \in S\}} > \epsilon$}
		\State $k \leftarrow k+1$
		\For{$\mathfrak{s} \notin Acc_\mathcal{Z}$}
		\State $\mathbf{Q}^{k+1}(\mathfrak{s}) \leftarrow\allowbreak
		\max_\mu\min_{\tau,\xi}\allowbreak\bigg\{\sum_{u_{C} \in U_{C}}\sum_{u_{A} \in U_{A}}\sum_{(s',q',\mathbf{v}',y) \in S_\mathcal{Z}}\allowbreak\tau((s,q,\mathbf{v},y),u_{A})\mathbf{Q}((s',q',\mathbf{v}',y'))\allowbreak Pr_\mathcal{Z}\left((s^{\prime},q^\prime,\mathbf{v}',y')|(s,q,\mathbf{v},y), u_C, u_A\right)\bigg\}$
		\EndFor
		\EndWhile
		\State \Return $\mathbf{Q}^k$
		\EndProcedure        	
	\end{algorithmic}
\end{algorithm}
%

\begin{theorem}\label{thm:convergence}
	Algorithm \ref{algo:value} converges in a finite number of iterations. Moreover, the value vector returned by Algorithm \ref{algo:value} is in an $\epsilon$-neighborhood of $\mathbf{Q}^*$.
\end{theorem}

%% file: simulation.tex
\section{Case Study}\label{Sec:CaseStudy}
In this section, we demonstrate the solution approach of Section \ref{Sec:SolutionApproach} on a signalized traffic network. 
The simulations were carried out using MATLAB$^{\tiny{\textregistered}}$ on a Macbook Pro with a 2.6GHz Intel Core i5 CPU and 8GB RAM. The Appendix contains an example on two-tank system.

\subsection{Signalized Traffic Network Model}
We consider signalized traffic network under the remote control of a transportation management center (TMC). A signalized traffic network consists of a set of links $\{1,2,\dots,L\}$ and intersections $\{1,2,\dots,N\}$ \cite{coogan2015traffic}. 
Each intersection can take a `red' signal which will not allow vehicles to pass through the intersection, or a `green' signal which will allow vehicles to pass. 
The number of vehicles in link $l$ at a time $k$ is $x_l(k)$ and $\Bar{x}_l$ denotes the capacity of link $l$. 
The number of vehicles entering the traffic network at a time $k$ is assumed to follow a Poisson distribution. 
Vehicles can travel through a link if and only if the subsequent intersection in the direction of travel is green. 
The link is then said to be \emph{actuated}. 
We assume that the flow rate $c_l$ of each link $l$ is given and fixed. 
The TMC is given an MITL objective that needs to be satisfied on the network. 
When the TMC issues a green signal at an intersection $n$, the turn ratio $\gamma_{ll'}\in[0,1]$ denotes the fraction of vehicles that will move to link $l'$ from link $l$ through intersection $n$. 
The maximum number of vehicles that can be routed to $l'$ from $l$ is determined by the supply ratio $\alpha_{ll'}$ of link $l'$, which is determined by the remaining capacity $\Bar{x}_{l'}-x_{l'}(k)$ of link $l'$. Given the above parameters, the dynamics of the link queues can be determined \cite{coogan2015traffic}.

We assume there is an adversary who can initiate actuator and timing attacks. An actuator attack will tamper with the traffic signal issued by the TMC. 
For instance, if the TMC actuates a link $l$ at time $k$ and the adversary attacks link $l$, then this link will not be actuated at time $k$. 
A timing attack will manipulate the timing information perceived by the TMC. Hence, any time stamped measurement $\{x_l,k\}$ received by the TMC indicating the number of vehicles at link $l$ at time $k$ might be manipulated to $\{x_l,k'\}$, where $k'$ is the time stamp that has been changed by the adversary.

The signalized traffic network model can be mapped to a DSG in the following way. States of the DSG are obtained by partitioning the number of vehicles on each link (e.g., box partition) \cite{coogan2015traffic}. The control action set at each intersection models which links can be actuated. The action set is then realized by taking the Cartesian product of the action sets at each intersection. 
The realized traffic signal at an intersection is jointly determined by the actions of the TMC and adversary. 
The transition and duration probability distributions between states are obtained from Algorithm \ref{alg:abstract}.

\subsection{Experimental Evaluation}

\begin{figure}[!t]
	\begin{center}
		\begin{tabular}{c}
			\scalebox{0.18}{\includegraphics[width=15in]{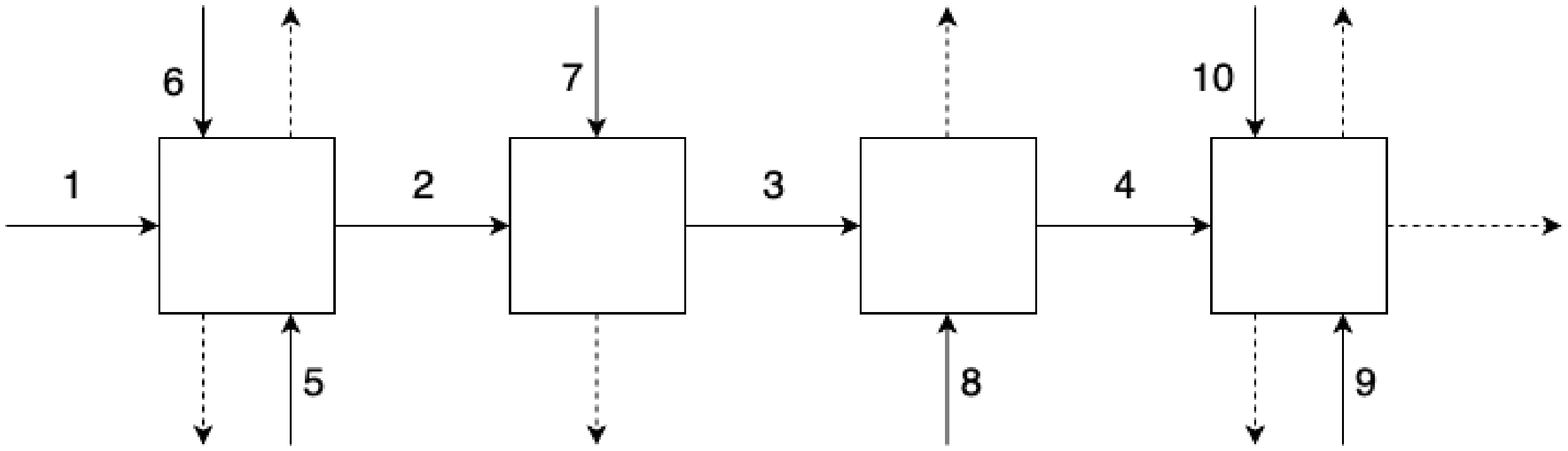}}
		\end{tabular}
		\caption{Representation of a \emph{signalized traffic network}. The network consists of $4$ intersections and $16$ links. Intersections are represented by squares, and links by arrows. Dotted arrows denote outgoing links that are not explicitly modeled.}
		\label{fig:network}
	\end{center}
\end{figure}

A representation of the signalized traffic network is shown in Fig. \ref{fig:network}. The network consists of $4$ intersections (squares) and $16$ links (arrows). 
We denote the intersections with incoming links $1$, $2$, $3$, and $4$ as intersections $1$, $2$, $3$, and $4$, respectively. The links represented by dotted arrows are not explicitly modeled \cite{coogan2015traffic}. For each intersection in Fig. \ref{fig:network}, the links that can be actuated by the TMC are given as follows:
\begin{itemize}
	\item Intersection $1$: $\{\{1\},\{5,6\}\}$;
	\item Intersection $2$: $\{\{2\},\{7\}\}$;
	\item Intersection $3$: $\{\{3\},\{8\}\}$;
	\item Intersection $4$: $\{\{4\},\{9,10\}\}$.
\end{itemize}
We assume that the TMC can actuate exactly one subset of links at each intersection so that no safety constraint will be violated. The link capacities are set to $\Bar{x}_1 = \dots = \Bar{x}_5 = 30$ and $\Bar{x}_6=\dots = \Bar{x}_{10} = 40$. Flow rates associated to each link are set to $c_1 = \dots = c_4 = 10$, $c_5 = \dots = c_{10} = 5$
 \cite{coogan2015traffic}. The supply ratios $\alpha_{ll'} = 1$ for all $l,l'$, 
 and the turn ratios are set to $\gamma_{12} = 0.3$, $\gamma_{23}=\gamma_{34}=\gamma_{52}=\gamma_{62}=\gamma_{73}=\gamma_{84}=0.5$. 
Vehicles entering a link in $(l_1,\dots, l_{10})$ follow a Poisson distribution with mean $(5, 0,0,0,5,5,0,0,5,5)$. We consider a time horizon of length $5$. 
The defender's strategy to detect a timing attack is to compare the deviation between its estimated and observed clock valuations with a pre-specified threshold $\mathbf{e}=2$. In particular, when $\|\boldsymbol{\lambda}-\mathbf{v}\|\leq 2$, hypothesis $\mathcal{H}_0$ holds and no timing attack is detected by the defender. When $\|\boldsymbol{\lambda}-\mathbf{v}\|> 2$, hypothesis $\mathcal{H}_1$ holds and an alarm indicating a timing attack is triggered. In this case, the FSC equipped by the controller has $5$ internal states.
\begin{table}[h!]
	\centering
	\caption{A sample sequence of the traffic light realized at each intersection for the MITL specification $\varphi_3=\Diamond_{[0,5]}\left( (x_2 \leq 10) \land (x_3 \leq 10) \land (x_4 \leq 10)\right)$. The letter `R' represents a `red' signal, and `G' represents `green' signal.}
	\begin{tabular}{|c |c| c| c| c|}
		\hline
	     &\multicolumn{4}{|c|}{Intersection} \\
	     \hline
	Time &  1 &  2 &  3 &  4\\ [0.5ex]
		\hline\hline
	 1 & G & R & G & R\\
	 2 & R & G & R & G\\
	 3 & G & G & R & G\\
	 4 & G & G & G & G\\
	 5 & G & R & G & G\\ [1ex]
		\hline
	\end{tabular}
	\label{table:signal}
\end{table}
\begin{figure}[!h]
	\begin{center}
		\begin{tabular}{c}
			\scalebox{0.18}{\includegraphics[width=15in]{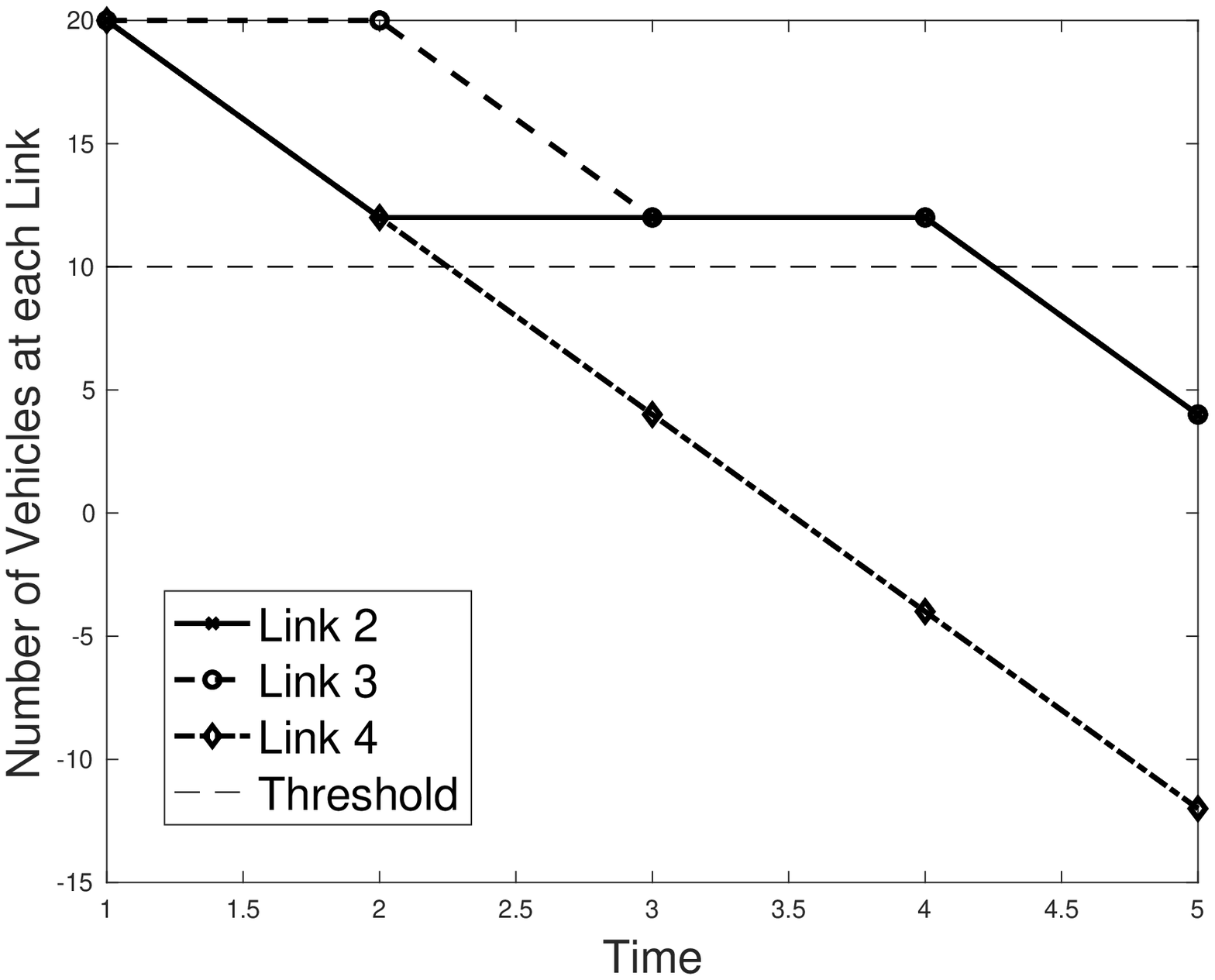}}
		\end{tabular}
		\caption{Number of vehicles on Links $2$, $3$, and $4$ at each time corresponding to the MITL formula $\varphi_3=\Diamond_{[0,5]}\left( (x_2 \leq 10) \land (x_3 \leq 10) \land (x_4 \leq 10)\right)$. In the presence of an adversary, the defender adopts an FSC-based policy with one realization shown in Table \ref{table:signal}. The dotted horizontal line is the threshold for the maximum number of vehicles allowed ($=10$). The three curves indicate that the number of vehicles in the links satisfies the MITL objective since they are each lower than $10$ before $5$ time units.}
		\label{fig:traffic}
	\end{center}
\end{figure}

The TMC is given one of the following MITL objectives.
\begin{enumerate}
	\item The number of vehicles at link $2$ is eventually below $10$ before deadline $d=5$: $\varphi_1=\Diamond_{[0,5]}(x_2\leq 10)$.
	\item The number of vehicles at link $2$ and $3$ are eventually below $10$ before deadline $d=5$: \\ $\varphi_2=\Diamond_{[0,5]}\left((x_2 \leq 10) \land (x_3 \leq 10) \right)$.
	\item The number of vehicles at link $2$, $3$ and $4$ are eventually below $10$ before deadline $d=5$: \\ $\varphi_3=\Diamond_{[0,5]}\left((x_2 \leq 10) \land (x_3 \leq 10) \land (x_4 \leq 10)\right)$.
\end{enumerate}
Our experiments yield the probability of satisfying each specification as: $\mathbf{\mathbb{P}(\varphi_1)=0.723}$, $\mathbf{\mathbb{P}(\varphi_2)=0.371}$, and $\mathbf{\mathbb{P}(\varphi_3)=0.333}$. These values agree with intuition since $\varphi_3$ being satisfied implies $\varphi_2$ holds true, which in turn implies that $\varphi_1$ is true.

We compare our approach for the objective $\varphi_3$ with two baselines. In the first baseline, the TMC issues periodic green signals for links $1,2,3$, and $4$ at intersections $1,2,3$, and $4$, respectively. In the second baseline, the TMC always issues green signals for links $1,2,3,4$ at intersections $1,2,3,4$.

For the two baseline scenarios, the TMC commits to deterministic strategies. The adversary's actuator attack strategies are as follows. In the first case, the adversary launches actuator attacks when the TMC issues a green signal, and does not attack when the TMC issues a red signal. As a result, the realized traffic signal will be red for all time at every intersection. In the second case, the adversary launches an actuator attack at every time instant. This results in the realized traffic signal being red for all time at each intersection. As a consequence, the number of vehicles in links $2,3$ and $4$ will reach their capacities and the links will be congested for the rest of the time horizon. Therefore, the probabilities of satisfying the MITL specification using the baselines are zero.

Table \ref{table:signal} shows a realization of the traffic signals when the defender adopts an FSC-based policy proposed in Section \ref{Sec:SolutionApproach}. 
Figure \ref{fig:traffic} shows the number of vehicles in each link for this realization. 
The graph indicates that the defender's policy is successful in ensuring that the MITL objective is satisfied. 
Moreover, if the adversary's timing strategy is such that when the difference in the manipulated and actual clock valuations is less than $2$ (the pre-specified threshold), it remains stealthy, even though this is not an explicitly specified goal. 

%
The construction of the DSG using Algorithm \ref{alg:abstract} takes $24.22$ seconds. 
The computation of the global DSG takes $367.9$ seconds. Given the global DSG, Algorithm \ref{algo:value} takes $644.8$ seconds to compute the defender's FSC.

%% file: relatedwork.tex
\section{Related Work}\label{Sec:RelatedWork}

Markov decision processes (MDPs) probabilistically represent transitions between states depending on an action taken by an agent.
Semi-Markov decision processes (SMDPs) \cite{jewell1963markov} are used to model Markovian dynamics where the time taken for transitions between states is a random variable. 
SMDPs have been typically used to analyze problems in production scheduling \cite{pinedo2012scheduling,ross2014introduction} and optimization of queues \cite{sennott2009stochastic,tijms2003first,stidham1993survey}. 

Stochastic games (SGs) generalize MDPs to the setting when there is more than one agent taking an action \cite{fudenberg1991game}. 
The satisfaction of an LTL formula for two-player SGs when the players had competing objectives was presented in \cite{niu2018secure,niu2019optimal}. 
In this paper, the authors synthesized a policy for the agent that maximizes the probability of satisfying LTL formula. However, this approach is not applicable to the case where the adversary can launch a timing attack.

Reactive synthesis\footnote{Reactive systems interact continuously with their environments. Reactive synthesis is the construction of a reactive system from a logical specification.} under TL constraints has been studied in \cite{bloem2012synthesis,alur2016compositional,kulkarni2018compositional,raman2015reactive}. 
These approaches typically consider a turn-based setting where the environment is viewed as an adversarial player. In comparison, this paper considers a stochastic environment with the defender and adversary taking their actions simultaneously. Moreover, the works on reactive synthesis assumes that the controller has complete knowledge of the environment, whereas this is not the case in our setting.

Although temporal logic frameworks like LTL have the ability to specify a broad class of system properties, a drawback is that they cannot be used to specify properties that have explicit timing constraints. 
This shortcoming is consequently extended to the FSA that is constructed to represent the LTL formula. 
Timed automata (TA) \cite{alur1994theory} extend FSAs by attaching finitely many clock constraints to each state. 
A transition between any two states will be influenced by the satisfaction of clock constraints in the respective states. 
There has been significant work in the formulation of timed temporal logic frameworks \cite{bouyer2017timed}. 
Metric interval temporal logic (MITL) \cite{alur1996benefits} is one such fragment that allows for the specification of formulas that explicitly depend on time. 
Moreover, an MITL formula can be represented as a TA \cite{alur1996benefits, maler2006mitl} that will have a feasible path in it if and only if the MITL formula is true. 

A parallel body of work proposed the incorporation of probabilities to a TA \cite{beauquier2003probabilistic} to yield a probabilistic timed automaton (PTA). 
Two-player SGs were used as an abstraction of the PTA to present tight bounds on the aforementioned probabilities in \cite{kwiatkowska2009stochastic}. 
Stochastic timed games, defined in \cite{bouyer2009reachability}, assumed two players choosing their actions deterministically, and the environment as a `half-player' whose actions were probabilistic. 
The existence of a strategy for one player such that the probability of reaching a set of states under any strategy of the other player and the randomness in the environment was shown to be undecidable in general. 
The authors of \cite{brazdil2010stochastic} studied a two-player SG and showed that it was not possible for a player to have an optimal strategy that guaranteed the `equilibrium value' against every strategy of the opponent. 
However, they showed the existence of an almost-sure winning strategy for one player against any strategy of the other player. 
This was not a Markovian policy, since it depended on not only the most recent state, but also on previous states. 
In all the papers mentioned here, the games were turn-based, and there was not a temporal logic formula that had to be satisfied. 

The satisfaction of an MITL formula in a motion-planning context was studied in \cite{karaman2008vehicle, liu2014switching, zhou2016timed}. 
However, these works were tailored for a single agent, and did not consider the presence of an adversary. Moreover, the analyses in \cite{karaman2008vehicle, liu2014switching, zhou2016timed} restrict their focus to MITL formulas with reachability accepting conditions. The treatment in this paper is broader in scope, and considers arbitrary MITL formulas. 

FSCs were used to simplify the policy iteration procedure for POMDPs in \cite{hansen1998improved}. 
This approach was extended to explicity carrying out the search for policies in the policy space in order to iteratively improve the FSC in \cite{hansen1998solving}. 
The satisfaction of an LTL formula (for a single agent) in a partially observable environment was presented in \cite{sharan2014finite}. 
This was extended to the setting with an adversary, who also only had partial observation of the environment, and whose goal was to prevent the defender from satisfying the LTL formula in \cite{bhaskar2019finite}. 
These treatments, however, did not account for the presence of timing constraints on the satisfaction of a temporal logic formula. 

%% file: conclusion.tex
\section{Conclusion and Future Work}\label{Sec:Conclusion}

We investigated the problem of synthesizing controllers for time critical CPSs under attack. We proposed durational stochastic games to capture the interaction between the defender and adversary, and also account for time taken for transitions between states. 
The CPS had to satisfy a time-dependent objective specified as an MITL formula. We used a timed automaton representation of the MITL formula, the DSG, and a representation of the defender policy as a finite state controller to synthesize defender policies that would satisfy the MITL objective under actuator and timing attacks carried out by the adversary. 
We evaluated our solution method on a representation of a signalized traffic network.

A formal characterization of spatial and temporal robustness in the presence of an adversary is a topic of future research. 
A second topic of interest is the computational complexity of the control synthesis procedure in adversarial settings. A potential method to reduce this complexity is to use a coarse discretization to synthesize a control strategy \cite{munos2002variable}, and then refine it for the states with relatively poor performance to compute an improved control strategy.
%

%% file: appendix.tex
\section*{Appendix} \label{Sec:Appendix}

This appendix presents a case-study demonstrating our approach on a two-tank system. We then provide proofs of the some of the results presented in earlier sections.

\subsection{Simulation: Two-Tank System}

We demonstrate our solution approach with simulations carried out on the control of a two-tank system \cite{yordanov2012temporal}. The system is described by $x(k+1)=Ax(k)+Bu(k)+w(k)$, where $x(k)=[x_1(k),x_2(k)]^T$, $u(k)$, and $w(k)$ are state variables representing water levels, control input representing the inflow rate, and stochastic disturbance at time $k$, respectively. The defender transmits a control signal $u_C(k)$ to the actuator through a wireless communication channel. We set the initial levels in the two tanks to $x(0)=[0.11,0.35]^T$.

The system is subject to an attack initiated by an intelligent adversary. The control signal $u(k)$ received by the actuator is compromised as $u(k)=u_C(k)+u_A(k)$ due to the actuator attack, where $u_C(k)$ and $u_A(k)$ correspond to signals sent by the defender and adversary \cite{zhu2011stackelberg}. Due to the timing attack, the time-stamped measurement $\{x,k\}$ indicating the water level at time $k$ is manipulated as $\{x,k'\}$, where $k'$ is the time stamp that has been modified by the adversary.

The state space (water levels in tanks) is partitioned into 49 rectangular regions, i.e., the water level in each tank is divided into $7$ discrete intervals with discretization resolution $0.1$, each representing a state of the DSG. The control and adversary signals are in the ranges $[0,5\times 10^{-4}]$ and $[0,2\times 10^{-4}]$, respectively \cite{yordanov2012temporal}. Control and adversary action sets are obtained by discretization of these sets of inputs. The disturbance $w(k)$ is zero mean i.i.d. Gaussian with covariance $1.5\times 10^{-5}$. The transition and duration probabilities are obtained by Algorithm \ref{alg:abstract}. This procedure took about $18$ seconds. 

The system needs to satisfy an MITL specification given by $\varphi=\Diamond_{[0,5]}\left(\bigvee_z\bigwedge_{i\in\{1,2\}}(z\leq x_i\leq z+0.1)\right)$, where $z\in\{0.3,0.4,0.5,0.6\}$. That is, before time $k=5$, the water levels in the two tanks should lie in the same discretization interval and are each required to be no less than $0.3$. If the MITL specification is satisfied, the difference between water levels in the two tanks should be at most $0.1$.

We compare our FSC-based policy with a baseline. The baseline does not account for the presence of the adversary. 
The results of our experiments are presented in Fig. 3. The baseline is evaluated for scenarios where the adversary is present and the adversary is absent. When there is no adversary, we observe that the baseline policy satisfies the MITL objective (the water levels in the tanks are $0.35$ and $0.30$, and the difference in the levels is $0.05$). However, when this policy is used in the presence of the adversary, we observe that the water level in the second tank falls below $0.3$, and the difference in the levels exceeds $0.1$, thereby violating the specification. This necessitates the use of an alternative control strategy for systems under attacks. Using our approach, we observe that the water levels in the two tanks are $0.33$ and $0.30$, and the difference in the levels is $0.03$. Moreover, these water levels are attained before the required deadline of $k=5$, which satisfies the MITL objective.

\subsection{Proofs}
%
\begin{proof}[Proof of Proposition \ref{Prop: Consistency}]
	For any transition in $\mathcal{P}$, $Pr\left((s^{\prime},q^\prime,\mathbf{v}')|(s,q,\mathbf{v}), u_C, u_A\right)\in[0,1]$. 
	This is due to the fact that $T_\mathcal{G}(\delta|s,u_C,u_A,s')\in[0,1]$ and $Pr_\mathcal{G}(s'|s,u_C,u_A)\in[0,1]$. 
	Moreover, \emph{i)} $Pr\left((s^{\prime},q^\prime,\mathbf{v}')|(s,q,\mathbf{v}), u_C, u_A\right)=0$ iff $T_\mathcal{G}(\delta|s,u_C,u_A,s')=0$, or $Pr_\mathcal{G}(s'|s,u_C,u_A)=0$, or both; 
	\emph{ii)} $Pr\left((s^{\prime},q^\prime,\mathbf{v}')|(s,q,\mathbf{v}), u_C, u_A\right)=1$ iff $T_\mathcal{G}(\delta|s,u_C,u_A,s')=1$ and $Pr_\mathcal{G}(s'|s,u_C,u_A)=1$. 
	Let $I_{(q,\mathbf{v}),(q',\mathbf{v}')}^\delta\coloneqq\mathbbm{1}((q,\mathbf{v})\xrightarrow{L(s'),\delta}(q',\mathbf{v}'))$ be an indicator function that takes value $1$ if its argument is true, and $0$ otherwise. 
	Then, Equation \eqref{eq:well defined prob} can be rewritten as:
	\begin{align}
		&\sum_{(s',q',\mathbf{v}')}T_\mathcal{G}(\delta|s,u_C,u_A,s')Pr_\mathcal{G}(s'|s,u_C,u_A)\label{eq:well defined prob deriv 1}\\
		&=\sum_{s'\in S_\mathcal{G}}\sum_{\delta\in\Delta}T_\mathcal{G}(\delta|s,u_C,u_A,s')I_{(q,\mathbf{v}),(q',\mathbf{v}')}^\delta\nonumber\\ &\quad\quad\quad\quad\quad\quad\quad\quad\quad Pr_\mathcal{G}(s'|s,u_C,u_A)\label{eq:well defined prob deriv 2}\\
		&=1,\label{eq:well defined prob deriv 3}
	\end{align}
	Equation \eqref{eq:well defined prob deriv 1} holds by substituting from Equation \eqref{eq:transition prob}, Equation \eqref{eq:well defined prob deriv 2} follows from Definition \ref{def:PDSG} 
	and $Pr_\mathcal{G}(s'|s,u_C,u_A)>0$, and Equation \eqref{eq:well defined prob deriv 3} results by observing that $\sum_{\delta\in\Delta}T_\mathcal{G}(\delta|s,u_C,u_A,s')=1$ and $\sum_{s'\in S_\mathcal{G}}Pr_\mathcal{G}(s'|s,u_C,u_A)=1$.
\end{proof}

\begin{proof}[Proof of Proposition \ref{proposition:correctness}]
	
	We prove the correctness of Algorithm \ref{algo:GAMEC} by first showing that no state or control action that belongs to a GEC will be removed. 
	Consider a GEC $(N,D)$.
	
	If there exists a state $\mathfrak{s}$ and control action $u_C\in U_C(\mathfrak{s})$ such that $Pr(\mathfrak{s}'|\mathfrak{s},u_C,u_A)>0$ for all $u_A$ and $\mathfrak{s}'\in N$, then according to lines 10 - 17 of Algorithm \ref{algo:GAMEC}, state $\mathfrak{s}$ and control action $u_C$ will not be removed since $D(\mathfrak{s})\neq\emptyset$. Therefore, Algorithm \ref{algo:GAMEC} never removes states or actions from a GEC.
	
	On the other hand, if there is a state $\mathfrak{s}\in N$ such that $D(\mathfrak{s})=\emptyset$, then $\mathfrak{s}$ will be removed (lines 10 - 17 of Algorithm \ref{algo:GAMEC}). Moreover, any state that can be steered to $\mathfrak{s}$ under some adversary action $u_A$ will also be removed (lines 18 - 26). Thus, any state or action that does not belong to GEC will be removed by Algorithm \ref{algo:GAMEC}, and the remaining states in $(N,D)$ after executions from lines 10 - 26 will form the GEC.
	
	Combining the arguments above, we have that Algorithm \ref{algo:GAMEC} computes a set of GECs $\{(C_i,D_i)\}_{i\geq 1}$ such that any GEC is contained by some $(C_i,D_i)$. Then by Definition \ref{def: GMEC} and line 35 of Algorithm \ref{algo:GAMEC}, we have that the result returned by Algorithm \ref{algo:GAMEC} is the set of GAMECs.
\end{proof}

\begin{proof}[Proof of Lemma \ref{lemma:convergence}]
	Given a control policy $\mu$, the GDSG $\mathcal{Z}$ is reduced to an MDP, $\mathcal{M}$. 
	Then, the composition of $M_\mu$ corresponds to a value iteration on $\mathcal{M}$. 
	The convergence of $M_\mu$ can be shown following the approach in \cite{bertsekas1995dynamic}. 
	
	Next, we show that the sequence obtained by composing $M$ is bounded and monotone.
	We observe that $M\mathbf{Q}(\mathfrak{s})$ is a convex combinations of all the neighboring states of $\mathfrak{s}$. 
	Moreover, $\mathbf{Q}(\mathfrak{s})\in[0,1]$ for all $\mathfrak{s}$, and is therefore bounded.
	We show that the sequence of value vectors is monotonically non-decreasing by induction. 
	Define $M^{-1}\mathbf{Q}:=\mathbf{0}$, and $M^{0}\mathbf{Q}(\mathfrak{s})=0$ for $\mathfrak{s} \notin \mathcal{E}$, and $M^{0}\mathbf{Q}(\mathfrak{s})=1$ for $\mathfrak{s} \in \mathcal{E}$. 
	Then, $M^{-1}\mathbf{Q}\leq M^{0}\mathbf{Q}$. 
	Suppose the sequence of value vectors is monotonically non-decreasing up to iteration $k$. We have
	\begin{align}
		&M^{k+1}\mathbf{Q}(\mathfrak{s})\nonumber\\
		\geq& \min_{\tau,\xi}\bigg\{\sum_{u_{C} \in U_{C}}\sum_{u_{A} \in U_{A}}\sum_{(s',q',\mathbf{v}',y) \in S_\mathcal{Z}}\nonumber\\
		&\tau((s,q,\mathbf{v},y),u_{A})\mathbf{Q}((s',q',\mathbf{v}',y'))\nonumber\\ 
		&Pr_\mathcal{Z}\left((s^{\prime},q^\prime,\mathbf{v}',y')|(s,q,\mathbf{v},y), u_C, u_A\right)\bigg\}\label{eq:monotone 1}\\
		\geq& \min_{\tau,\xi}\bigg\{\sum_{u_{C} \in U_{C}}\sum_{u_{A} \in U_{A}}\sum_{(s',q',\mathbf{v}',y) \in S_\mathcal{Z}}\nonumber\\
		&\tau((s,q,\mathbf{v},y),u_{A})\mathbf{Q}((s',q',\mathbf{v}',y'))\nonumber\\ 
		&Pr_\mathcal{Z}^{k-1}\left((s^{\prime},q^\prime,\mathbf{v}',y')|(s,q,\mathbf{v},y), u_C, u_A\right)\bigg\}\label{eq:monotone 2}\\
		=& M^k\mathbf{Q}(\mathfrak{s}),\label{eq:monotone 3}
	\end{align}
	where $Pr_\mathcal{Z}^{k-1}\left((s^{\prime},q^\prime,\mathbf{v}',y')|(s,q,\mathbf{v},y), u_C, u_A\right)$ is obtained by substituting $\mu^{k-1}(g',u_C|g,\mathbf{v})$ into \eqref{eq:GDSG transition prob}, inequality \eqref{eq:monotone 1} holds since $M^{k+1}\mathbf{Q}$ corresponds to a maximizing policy $\mu^{k+1}$, \eqref{eq:monotone 2} holds by induction, and \eqref{eq:monotone 3} follows from the construction of $\mu^k$. 
	Therefore, $\mathbf{Q}^{k+1}\geq\mathbf{Q}^k$, implying the sequence of value vectors is monotonically non-decreasing.
	From the boundedness and monotonocity of $M\mathbf{Q}$, the sequence is a Cauchy sequence that converges to a value $\mathbf{Q}^*$.
\end{proof}

\begin{proof}[Proof of Theorem \ref{thm:convergence}]
	We prove convergence by by showing that the sequence of value vectors computed in Algorithm \ref{algo:value} is bounded and monotonically non-decreasing. 
	Line 4 of Algorithm \ref{algo:value} serves as our induction base, i.e., $\mathbf{Q}^1\geq \mathbf{Q}^0$. 
	Line 8 of Algorithm \ref{algo:value} is equivalent to computing $\mathbf{Q}^{k+1}$ as $\mathbf{Q}^{k+1}=M\mathbf{Q}^k$. From Lemma \ref{lemma:convergence}, $\mathbf{Q}^{k+1}\geq \mathbf{Q}^k$. Convergence follows from the Monotone Convergence theorem \cite{royden2010real}.
	%
	That the control policy is within an $\epsilon$-neighborhood of SE follows from Line 5 of Algorithm \ref{algo:value}.
\end{proof}
